\documentclass[sigconf]{acmart}
\AtBeginDocument{%
  \providecommand\BibTeX{{%
    \normalfont B\kern-0.5em{\scshape i\kern-0.25em b}\kern-0.8em\TeX}}}

\usepackage{lineno,hyperref}
\usepackage{graphicx}
\usepackage{kotex}
\usepackage{multirow}
\usepackage{framed}
\usepackage{soul}
\usepackage{subcaption}
\usepackage{lipsum}
\usepackage{comment}
\usepackage{tablefootnote}
\usepackage{mathtools}
\usepackage{amsmath,amsfonts,amsthm}

\usepackage{multirow}
\usepackage{enumitem} 
\usepackage{amsmath}
\usepackage{subcaption}
\usepackage{footnote}
\usepackage{caption}
\usepackage{tikz}
\usepackage{stfloats} 
\usepackage{thmtools}
\usepackage{pifont}
\usepackage[linesnumbered,boxruled,ruled,noend]{algorithm2e}

\newtheorem{theorem}{Theorem}
\newtheorem{problemDefinition}{Problem definition}
\newtheorem{definition}{Definition}
\newtheorem{property}{Property}

\newtheorem{strategy}{Strategy}
\newtheorem{example}{Example}

\newcommand{\FBA}{\textsf{FBA}}
\newcommand{\FCA}{\textsf{FCA}}
\newcommand{\DFBA}{\textsf{DFBA}}

\newcommand{\spara}[1]{\smallskip\noindent{\bf #1}}

\DeclareMathOperator*{\argmin}{argmin} 
\DeclareMathOperator*{\argmax}{argmax} 

\setcopyright{acmcopyright}
\copyrightyear{2018}
\acmYear{2018}
\acmDOI{XXXXXXX.XXXXXXX}

\acmConference[Conference acronym 'XX]{Make sure to enter the correct
  conference title from your rights confirmation emai}{June 03--05,
  2018}{Woodstock, NY}
\acmPrice{15.00}
\acmISBN{978-1-4503-XXXX-X/18/06}

\begin{document}

\title{Effective and Efficient Core Computation in Signed Networks}

\author{Junghoon Kim}
\affiliation{%
  \institution{UNIST}
  \country{South Korea}}
\email{junghoon.kim@unist.ac.kr}

\author{Sungsu Lim}
\affiliation{%
  \institution{Chungnam National University}
  \country{South Korea}}
\email{sungsu@cnu.ac.kr}

\author{Jungeun Kim}
\affiliation{%
  \institution{Kongju National University}
  \country{South Korea}
}
\email{jekim@kongju.ac.kr}

\renewcommand{\shortauthors}{Kim et al.}

\begin{abstract}
With the proliferation of mobile technology and IT development, people can use social network services at any place and anytime. Among many social network mining problems, identifying cohesive subgraphs attract many attentions from different fields due to its numerous applications. Among many cohesive subgraph models, $k$-core is the most widely used model due to its simple and intuitive structure. 
In this paper, we formulate $(p,n)$-core in signed networks by extending $k$-core. $(p,n)$-core simultaneously guarantees the sufficient internal positive edges and deficient internal negative edges. We formally prove that finding an exact  $(p,n)$-core is NP-hard. Hence, we propose three efficient and effective algorithms to find a solution. Using real-world and synthetic networks, we demonstrate the superiority of our proposed algorithms. 
\end{abstract}

\maketitle

\section{Introduction}

With the recent rapid and prolific  development of mobile and communication technology, people can use online social networking services at any place and anytime.
Understanding online social networks helps us  understand human complex relationships; thus, many researchers are trying to analyse social networks by capturing their characteristics~\cite{scott1988social}. 
There are several interesting properties in social networks including degree distribution~\cite{barabasi1999emergence}, community structure~\cite{girvan2002community}, node influence~\cite{travenccolo2008accessibility},  small diameter~\cite{watts1998collective} and so on. Among them, mining of \textit{cohesive subgraphs} recently has received much attention from different fields due to its attractive applications even though there is no formal definition of a cohesive subgraph.

These days, many social networks consist of a social network and meta-data information named attributes, such as user or link information. These social network with attribute information are called \textit{attributed social networks}~\cite{chunaev2020community}. 
Attributed social networks include location-based social networks~\cite{kim2020densely}, keyword-based social networks~\cite{fang2016effective}, event-based social networks~\cite{feng2014search}, multi-layer graphs~\cite{kim2015community}, and uncertainty graphs~\cite{adar2007managing}, etc. In this study, we study a signed network~\cite{tang2016survey}. A signed network contains a set of nodes and edges in which each edge has either a positive (\textquotesingle+\textquotesingle) or a negative (\textquotesingle-\textquotesingle) sign. This sign can be interpreted as the relationship between two people, that is, a positive edge indicates that two users (or entities) are close; a negative edge implies that two users (or entities) are not in a good relationship. These signed networks have many applications including social psychology~\cite{antal2006social}, physics~\cite{facchetti2011computing}, correlation clustering~\cite{bansal2004correlation}, network robustness~\cite{Dey2021Network}, and link recommendation~\cite{tang2015negative}. 

In this paper, we study the cohesive subgraphs discovery problem~\cite{seidman1983network} in signed networks. Specifically, we extend the classic $k$-core on signed networks by considering two distinct edge types. 
The $k$-core is a representative cohesive subgraph model and is  widely used for social network analysis. It is utilised for graph clustering~\cite{cheng2013local}, graph visualisation~\cite{alvarez2006large}, community search~\cite{sozio2010community}, and identifying influential nodes~\cite{bae2014identifying}.  
There are many variations~\cite{malliaros2020core} of the $k$-core in different domains such as distance-generalised $k$-core~\cite{bonchi2019distance}, bipartite $(\alpha,\beta)$-core~\cite{ding2017efficient}, radius-bounded cores~\cite{wang2018efficient} and so on. Recently, Giatsidis et al.~\cite{giatsidis2014quantifying} extended classic $k$-core for directed signed networks. However, the proposed problem~\cite{giatsidis2014quantifying} does not consider the internal negative edges of the resultant subgraphs. It only guarantees sufficient incoming internal positive edges and sufficient outgoing external negative edges. Hence, \cite{giatsidis2014quantifying} cannot be directly utilised for some applications such as team formation or group recommendation because the resultant subgraphs may contain many negative edges among users. This implies that the identified subgraph contains many negative internal edges; thus, this indicates that the quality of the resultant subgraph is not guaranteed. 

Instead of extending $k$-core, several approaches~\cite{wu2020maximum,li2018efficient} for finding cohesive subgraphs in signed networks have been proposed, including $k$-truss based~\cite{wu2020maximum} and clique based~\cite{li2018efficient} approaches. Even if both approaches have more cohesiveness compared with $k$-core~\cite{fang2020survey}, selecting the proper parameter is very challenging to end-users. For example, to find an appropriate parameter $k$ for a $k$-truss model, the end-users may need to understand the concept of the triangle and adjacent triangles. Similarly, the clique-based method also requires users to understand the model in order to select appropriate parameters. On the other hand, 
our extended $k$-core based model is very simple and intuitive to end-users; thus, selecting appropriate parameters is relatively easier than other approaches. 

In this paper, by preserving the simplicity and intuitive structure of the $k$-core in networks, we propose a new signed core problem named $(p,n)$-core in signed networks. 
$(p,n)$-core inherits the traditional $k$-core. More specifically, $(p,n)$-core aims to find a maximal subgraph while guaranteeing sufficient internal positive edges and deficient internal negative edges of a resultant subgraph. 

Our problem has many applications by directly extending the applications of the classic $k$-core such as 
\textit{network analysis}~\cite{seidman1983network}, 
\textit{community search}~\cite{sozio2010community}, 
\textit{visualisation}~\cite{alvarez2006large},  
\textit{similarity}~\cite{nikolentzos2018degeneracy}, and
\textit{anomaly detection}~\cite{shin2018patterns}. 

For example, $k$-core can be utilised for the community search problem. The community search problem is defined as follows: given a graph $G$ and a set of query nodes $Q$, it aims to find a community $C$ containing all the query nodes $Q$ while $C$ is connected and satisfies some constraints. Even if there are several community search models, the $k$-core based CS (community search) model is most widely used due to its simple and intuitive structure. It finds a maximal subgraph containing all the query nodes while the minimum degree is maximised~\cite{sozio2010community}. We can simply extend $(p,n)$-core for the community search problem by finding a connected component that contains all the query nodes while maximising the positive edges by $p$ or minimising the negative edge by $n$. In Section~\ref{sec:experiments}, we present the result of the community search by utilising our problem.

In summary, we focus on the $k$-core problem in signed networks by unifying two key concepts. Each component is as follows. 

\begin{enumerate}[noitemsep,nolistsep,leftmargin=*]
    \item \textit{Considering positive and negative edge constraints}: We consider both positive and negative edge constraints to find $(p,n)$-core to guarantee the sufficient internal positive edges and deficient internal negative edges.
    \item \textit{Maximality} : We find a maximal subgraph as a result. In Section~\ref{sec:problem}, we present that the result of $(p,n)$-core is not unique. 
\end{enumerate}

\subsection{Challenges}
Since finding an exact $(p,n)$-core  is NP-hard (See Section~\ref{sec:problem}), we have two primary challenges: 
\begin{enumerate}
    \item \textit{Challenge 1 : effectiveness}: how do we identify the effective solution?
    \item \textit{Challenge 2 :  efficiency}: how do we identify scalable (efficient) solutions?
\end{enumerate}
To handle both challenges, we propose three algorithms which are summarised in Table~\ref{tab:sum}. 
Note that the notation \textit{followers} implies a set of nodes that would be deleted together when we remove a node. 
\begin{table}[ht]
\caption{Summary of proposed algorithms}
\label{tab:sum}
\centering
\begin{tabular}{c|c|c|c}
\hline
              & {\FBA} & {\DFBA} & {\FCA} \\ \hline \hline
effectiveness &  $\bigstar $   &  $\bigstar \bigstar \bigstar \bigstar$   & $\bigstar $ \\ \hline
efficiency    &  $\bigstar $   &  $\bigstar \bigstar$   &  $\bigstar \bigstar \bigstar \bigstar$    \\ \hline \hline
\end{tabular}
\end{table}

A follower-based algorithm ({\FBA}) is designed to compute the exact size of the followers after computing the candidate nodes. Even if {\FBA} puts effectiveness as a major concern, it fails to preserve effectiveness and efficiency since it does not consider the importance of the nodes in the negative graph and  is required to compute all the followers of the nodes.     

Next, we propose a new improved {\FBA} named a disgruntled follower-based algorithm (\DFBA) by improving efficiency and effectiveness. For {\DFBA}, we have two strategies: (1) \textit{Propose pruning strategy} : to avoid computing all the sizes of the followers, we propose a new pruning strategy to improve efficiency; and (2) \textit{Set a new objective function} : by considering the sizes of the followers and the node importance to satisfy the negative edge constraint, we define a new objective function to select a node to be removed.

Even if {\DFBA} improves effectiveness and efficiency, it still cannot handle large-scale networks since it intrinsically computes many followers to pick the best node. Thus, we propose a fast circle-based approach ({\FCA}).

\subsection{Contributions}
The contributions of this work are summarized as follows:
\begin{itemize}
    \item \textit{Problem definition :} To the best of our knowledge, this is the first work for $k$-core in signed networks considering sufficient internal positive edges and deficient internal negative edges;
    \item \textit{Theoretical analysis :} We prove that finding an exact solution of the $(p,n)$-core problem is NP-hard and the solution is not unique; 
    \item \textit{Designing new algorithms :} We propose three algorithms to identify the high-quality solution of $(p,n)$-core; and
    \item \textit{Extensive experimental study :} By using  real-world and synthetic networks, we conduct extensive experiments to demonstrate the superiority of our proposed algorithms. 
\end{itemize}

\section{$(p,n)$-core in signed networks}\label{sec:problem}

A signed network is modelled as a graph $G = (V, E^+, E^-)$ with a set of nodes $V$, a set of positive edges $E^+$, and a set of negative edges $E^-$. We denote a positive graph (or a negative graph) $G^+$ (or $G^-$) to represent the induced subgraph consisting of positive (or negative) edges, and we denote $V^+$ (or $V^-$) to represent a set of nodes in the positive or negative graph. In this paper, we consider that the graph $G$ is unweighted and undirected. Given a set of nodes $C\subseteq V$, we denote $G[C]$ as the induced subgraph of $G$ which takes $C$ as its node set and $E[C]=\{(u,v) \in E|u,v \in C\}$ as its edge set. Note that for any pair of nodes, we allow them to have positive and negative edges simultaneously. Positive edges indicate a kind of public relationship, and  negative edges indicate a private negative relationship. Many social network services allow hiding other users' posts even if they are friends. 
To introduce our problem, we present some basic definitions used in this paper, which are explained in Table~\ref{tab:notation}.

\begin{figure}[t]
\centering
\includegraphics[width=0.99\linewidth]{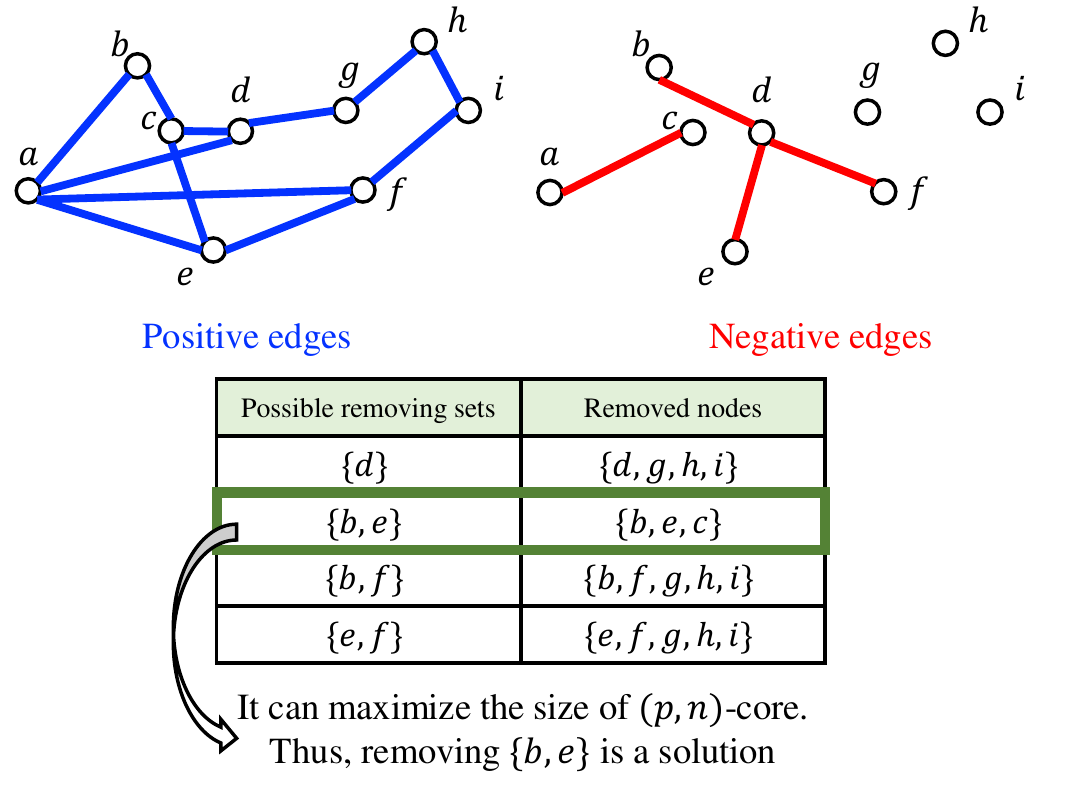}
\vspace{-0.3cm}
\caption{Example of $(p,n)$-core}
\label{fig:example}
\end{figure}

\begin{table}[t]
\caption{Notations}
\label{tab:notation}
\centering
\begin{tabular}{c|c}
\hline
\multicolumn{1}{c|}{\textbf{Notation}} & \multicolumn{1}{c}{\textbf{Description}}   \\ \hline \hline
$G=(V,E^+, E^-) $                             & A signed network                   \\ \hline
$p $                             & positive edge constraint           \\ \hline
$n$                              & negative edge constraint           \\ \hline
$H=(V_H, E_H^+, E_H^-)$                              & a subgraph $H$ of $G$           \\ \hline
$H^+, H^-$                              & a positive (or negative) graph $H$ of $G$           \\ \hline
$D(u)$                          & disgruntlement of the node $u$     \\ \hline
$F(u)$                           & followers of node $u$          \\ \hline
$\tilde{F}(u)$                           & approximated number of followers of node $u$          \\ \hline
$LB(u)$                           & lower-bound of the node $u$'s followers  \\ \hline
$UB(u)$                           & upper-bound of the node $u$'s followers  \\ \hline
\hline
\end{tabular}
\end{table}

\begin{definition}
(\underline{Signed network}). 
A signed network $G=(V,E^+, E^-)$ consists of $V$ nodes, a set of positive edges $E^+$ and a set of negative edges $E^-$. 
\end{definition}

Figure~\ref{fig:example} depicts an example of a signed network. Note that the blue-coloured edges indicate positive edges and red-coloured edges indicate negative edges. This signed network contains $9$ nodes, $12$ positive edges, and $4$ negative edges. Next, we present two constraints to define the $(p,n)$-core problem.  

\begin{definition}
(\underline{Positive edge constraint}). 
Given a signed network $G=(V,E^+, E^-)$ and positive integer $p$, a subgraph $H \subseteq V$ satisfies the positive edge constraint if any nodes in $H$ have at least $p$ positive edges in $H$, i.e., $\delta(H)\geq p$. 
\end{definition}

\begin{definition}
(\underline{Negative edge constraint}). 
Given a signed network $G=(V,E^+, E^-)$ and integer $n$, a subgraph $H \subseteq V$ satisfies the negative edge constraint if any nodes in $H$ have less than $n$ negative edges in $H$, i.e., $\gamma(H) < n$.
\end{definition}

Note that both edge constraints are utilised to define our problems.
Next, we present $k$-core which is related to the positive edge constraint. 

\begin{problemDefinition}\label{def:kcore}
(\underline{$k$-core}~\cite{seidman1983network}). 
Given a graph $G=(V,E)$ and a positive integer $k$, $k$-core, denoted as $H$, is a maximal subgraph of which each node has at least $k$ neighbours in $H$.  
\end{problemDefinition}

It is known that finding $k$-core in a network can be achieved in a polynomial time~\cite{batagelj2003m}.
In the following, we use a notation $p$-core graph of a signed network which indicates that the induced subgraph of $p$-core, i.e., $H=G[D]$ where $D$ is a $p$-core. 
Now, the $(p,n)$-core problem is ready to be formulated.

\begin{problemDefinition}
(\underline{$(p,n)$-core}). 
Given a signed network $G$, positive integer $p$, and $n$, $(p,n)$-core, denoted as $H$, is a subgraph of $G$ where every node satisfies the positive and negative edge constraints in $H$, i.e., we aim to find a subgraph $H$ such that
\begin{itemize}
    \item $|H|$ is maximised;
    \item $\delta(H) \geq p$;
    \item $\gamma(H) < n$
\end{itemize}
\end{problemDefinition}

Note that the definition of $(p,n)$-core is an extension of $k$-core by additionally considering the negative edge constraint.

\begin{example}
Figure~\ref{fig:example} shows an example of a signed network. Suppose that $p=2$ and $n=2$. 
In this case, the whole graph does not satisfy the negative edge constraint since the node `$d$' has three negative edges. There are four options to satisfy the negative edge constraint.
First, the node $d$ can be removed, then the constraint can be satisfied. However, the size of the remaining graph is $5$ since the nodes $\{d, g, h, i\}$ are removed together due to the positive edge constraint.  
Secondly, the nodes $\{b,e\}$ can be removed. Then the node $c$ is additionally removed due to the positive edge constraint; therefore, the size of the remaining graph is $6$. 
Thirdly, the nodes $\{b,f\}$ can be removed. Then  the nodes $\{b,f,g,i,h\}$ are removed cascadingly, and the size of the remaining graph is $4$.  
Lastly, the nodes $\{e,f\}$ can be removed. The nodes $\{b,f,g,i,h\}$ are removed together, and then the size of the remaining graph is $4$. 
Since  maximising the size of $(p,n)$-core is our goal, removing the nodes $\{b,e\}$ is the best choice to get a solution. 
\end{example}

\subsection{Properties of $(p,n)$-core}

\begin{property}
$(p,n)$-core is not unique. 
\end{property}

\begin{proof}
Suppose that we have two solutions of $(p,n)$-core named $D_1$ and $D_2$ while $|D_1|=|D_2|$. A subgraph $D^*=D_1\cup D_2$ is then formed. Here, two different node sets are defined. 
(1) $D_{1,2} = D_1 \setminus D_2$; and
(2) $D_{2,1} = D_2 \setminus D_1$. 
There must be a set of negative edges from $D_{1,2}$ to $D_{2,1}$; otherwise, $D^*$ can be the solution. This implies that removing $D_{1,2}$ makes the nodes $D_{2,1}$ satisfy the negative edge constraint. 
For example, suppose that $|D_{1,2}|=|D_{2,1}|=1$. Then, two nodes $a\in D_{1,2}$ and $u\in D_{2,1}$ can be obtained. In $D^*$, nodes $a$ and $u$ are connected in the negative graph, and they have exactly $n$ negative edges. Otherwise, both $D_1$ and $D_2$ cannot be the solution. This implies that if two nodes $a$ and $u$ are connected in the negative graph and they have $n$ negative edges, $(p,n)$-core can have two solutions. Therefore, the solution of $(p,n)$-core is not unique. 
\end{proof}

\begin{theorem}
Finding a solution of $(p,n)$-core is NP-hard.
\end{theorem} 

\begin{proof}
To show the hardness, we utilise the $k$-clique problem which is a classic problem and NP-hard. 
First, suppose that we have an instance of $k$-clique : $I_{KC}=(G=(V,E), k)$. We then show a reduction from $I_{KC}$ to an instance of our problem. We first construct a signed network $S=(V, E^+, E^-)$ where $E^+ = E$, and we generate $|V|(|V|-1)$ negative edges in $E^-$, and then $n$ and $p$ are set to $k+1$ and $k-1$, respectively. 
It implies that the size of the solution must be larger than or equal to $k$, and less than $k+1$, i.e., the size is $k$ and its minimum degree is $k-1$. This indicates that finding a solution of  $I_{pncore}=(S, k-1, k+1)$ is to  find a $k$-clique in $I_{KC}$. As mentioned above, a reduction from an instance $I_{KC}$ to the instance $I_{pncore}$ can be done in polynomial time. 
Therefore, the $(p,n)$-core problem is NP-hard for finding an exact solution. 
Figure~\ref{fig:np} shows a reduction from the $k$-clique problem to our $(p,n)$-core problem.  

\begin{figure}[t]
\centering
\includegraphics[width=0.99\linewidth]{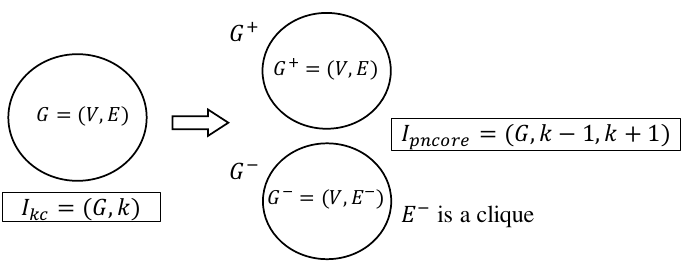}
\vspace{-0.2cm}
\caption{Reduction procedure}
\label{fig:np}
\end{figure}
\end{proof}

\begin{property}
If any solution of $(p,n)$-core where $p\geq 2$ and $n\geq 2$ contains a node $u$, any solution of ($p-1,n$)-core must contain the node $u$ and any solution of ($p,n+1$)-core must contain the node $u$. 
\end{property}

\begin{proof}
This property implies that a solution of $(p,n)$-core can be a solution of  ($p-1,n$)-core or ($p,n+1$)-core. A node of a solution in $(p,n)$-core guarantees that the node has at least $p$ positive edges. This directly implies that it has at least $p-1$ positive edges. Similarly, a node of a solution in $(p,n)$-core indicates that the node has less than $n$ negative edge. Thus, the node belongs to ($p,n+1$)-core. 
\end{proof}

In this paper, we consider a signed network as the two-layer graph. Figure~\ref{fig:example} shows a two-layer graph consisting of a positive graph layer and a negative graph layer for further clarity.

\section{Algorithms for $(p,n)$-core}\label{sec:algorithms}

In this section, we present three algorithms with different strategies for finding a solution of $(p,n)$-core. 
First, we propose a follower-based algorithm named {\FBA} which puts a positive edge constraint as the primary concern. Next, we propose {\DFBA} which puts both positive and negative edge constraints as a major concern simultaneously. However, both approaches must compute a set of nodes which is removed together (named \textit{followers}) for every iteration to find the best nodes to be removed. Hence, we propose a fast-circle algorithm, named {\FCA}, by approximately estimating the size of the followers to select the best node to be removed. 
In this section, to avoid confusion, we interpret a signed network as a two-layer graph (the positive graph layer and negative graph layer can be checked in Figure~\ref{fig:example}).

\subsection{Follower-Based Algorithm (\FBA)}\label{sec:FBA}
We first present an algorithm named {\FBA} which focuses on the positive edge constraint as the major concern. As our goal is to identify a set of nodes satisfying the constraints, it is reasonable to remove a set of nodes to satisfy the constraints, as discussed in Section~\ref{sec:problem}. 
When a node is removed, a set of nodes can be removed together in a cascading manner due to the positive edge constraint (minimum degree constraint). This is because when a node is removed, the degree of the neighbour nodes decreases, which can make the nodes to not satisfy the positive edge constraint. Thus, we firstly define a concept named \textit{follower}. 

\begin{definition} \label{def:follower}
(\underline{Followers}). Given an integer $p$, positive graph $H$ with $\delta(H)\geq p$, and node $v\in H$, followers of node $v$, denoted as, $F(v)$, consist of (1) node $v$ and (2) the nodes that are deleted cascadingly owing to the positive edge constraint with $p$ when node $v$ is removed.
\end{definition}

An example and properties are described as follows. 

\begin{example}
In Figure~\ref{fig:example}, suppose that $p=2$. The followers of some nodes are as follows : 
$F(a)=\{a,b\}$,
$F(b)=\{b\}$,
$F(c)=\{b,c\}$,
$F(d)=\{d,g,h,i\}$,
$\cdots$, and 
$F(i)=\{g,h,i\}$. 
Note that followers of a specific node $v$ contain the node $v$ based on  Definition~\ref{def:follower}. 
\end{example}

\begin{property}\label{property:pfa_lemma}
Given a $p$-core graph $H$, when we remove a node $u\in H^+$ with $|F(u)|=1$, the only node $u$ in the negative layer graph is removed, i.e., no other nodes are additionally removed. 
\end{property}

Property~\ref{property:pfa_lemma} is reasonable since the negative edge constraint does not require any cascading removal. 

\begin{property}\label{property:nfa_lemma}
When we remove a node $u$ in the negative layer graph, a set of nodes which contains the node $u$ can be removed together due to the positive edge constraint. 
\end{property}

Due to Property~\ref{property:pfa_lemma} and Property~\ref{property:nfa_lemma}, removing multiple nodes is determined in the positive graph. Thus, in {\FBA}, it is required to compute the followers of the candidate nodes to be removed. The strategy is as follows. 

\begin{strategy}
A set of nodes is identified in the negative graph of which each has at least $n$ negative neighbour nodes. We call these nodes key nodes. 
Next, a set of candidate nodes to be removed is found. The candidate nodes are the union of the key nodes and neighbours of the key nodes.
All the followers are computed for every candidate node, and then a node which has the minimum size of the followers is found. 
The node and its followers are then removed. 
This process is repeated until the remaining graph satisfies the positive and negative constraints simultaneously.
\end{strategy}

Note that when computing the followers of the candidate nodes, a set of nodes to be removed can be found by simply checking the size of the followers. 
Algorithm~\ref{alg:FBA} presents the {\FBA} procedure.
Initially, $p$-core and a set of candidate nodes are computed (lines 1-3). 
Next, all the followers of the candidate nodes are identified then the node having the smallest followers is selected and then removed (lines 5-9). This process is repeated until the negative edge constraint of the remaining graph is satisfied (lines 4-9). Finally, a set of nodes in the current subgraph is returned as a result (line 10).

\begin{algorithm}[t]
\SetKw{return}{return}
\SetKwData{D}{disgruntlement}
\SetKwData{T}{T}
\SetKwFunction{pcore}{$p$-core}
\SetKwFunction{pGraph}{$p$Graph}
\SetKwFunction{nGraph}{$n$Graph}
\SetKwFunction{computeFollower}{computeFollower}
\SetKwFunction{connectedComponent}{neighbours}
\SetKwFunction{removeNode}{removeNode}
\SetKwInOut{Input}{input}
\SetKwInOut{Output}{output}
\Input{Signed graph $G$, and positive integers $p$ and $n$}
\Output{$(p,n)$-core}
$H\leftarrow$ $G[$\pcore{$G$}$]$\;
\While{$\gamma(H^-) < n$}{
 $S \leftarrow \bigcup_{v\in H^-} v $ if $N(v, H^-) \geq n$\;
 $\T \leftarrow \bigcup_{s\in S} N(s, H^-)$\;
 $F(v) \leftarrow$ \computeFollower{$v, H^+$}, $\forall v\in \T$\;
 $u \leftarrow $ $\argmin_{v\in \T} |F(v)|$\;
 $H^+ \leftarrow$ \removeNode{$H^+, F(u)$}\; 
 $H^- \leftarrow$ \removeNode{$H^-, F(u)$}\; 
 $T \leftarrow T \setminus F(u)$\;
}
\return $V_{H}$\;
\caption{Follower-Based Algorithm (\FBA)}
\label{alg:FBA}
\end{algorithm}

\spara{Complexity.} The time complexity of {\FBA} is $O(|V|^2 (|V|+|E^+|))$ where $|V|(|V|+|E^+|)$ is to compute the followers of all the nodes in each iteration and the maximum number of iterations is $|V|$. 

\spara{Limitation of {\FBA}.} Even if {\FBA} is a straightforward approach, it does not consider the negative edges to select the best node to be removed. It only utilises the negative edges to select the candidate nodes. Therefore, it may remove many undesired nodes to satisfy the constraint, that is, it suffers from the effectiveness issue. In addition, for every iteration, it must compute all the followers in the remaining graph; thus, it also suffers from the efficiency issue.

\subsection{Disgruntled Follower-Based Algorithm (DFBA)}\label{sec:DFBA}

The previous sections discuss the {\FBA} and its limitations since it only focuses on the followers of the identified candidates to find proper nodes to be removed and computing the exact followers is required in every iteration. 

To incorporate the negative edge constraint, we introduce a definition named \textit{disgruntlement} which is an indicator that describes how helpful it is to satisfy the negative edge constraint when we remove a node. Note that a node having a large disgruntlement is preferred to be removed since it helps satisfy the negative edge constraint. 
Therefore, in this section, we simultaneously consider both followers and disgruntlement to maximise the size of the resultant $(p,n)$-core. Since a large disgruntlement and a small number of followers is preferred, we aim at maximising the following function. 

\begin{align}\label{eq:main_obj} 
O(.) = \frac{D(.)}{|F(.)|}
\end{align}

In this section, we design an algorithm by iteratively removing a node which has the maximum $O(.)$. However, there is a major concern when computing the Equation~\ref{eq:main_obj}: \textit{how can $|F(.)|$ be computed efficiently?} Given a node $v$, computing $F(v)$ requires $O(|V|+|E^+|)$. For all the nodes,  computing the value $F(.)$ is required, and should be for every iteration. Note that the maximum iteration is $|V|$. To handle this issue, we develop a lower-bound(Section~\ref{sec:lowerbound}) and an upper-bound(Section~\ref{sec:upperbound}), and present an algorithmic procedure (Section~\ref{sec:DFBA_algorithmic}). 

\subsubsection{Considering negative edges}
Here, a disgruntlement is formally defined and an example of computing a disgruntlement is shown. Given a negative graph $G^-=(V^-, E^-)$ and a negative edge constraint $n$, a node $u\in V^-$ is denoted as a key node if its degree is larger than or equal to $n$. 

\begin{definition}
(\underline{Disgruntlement}). 
Given a negative graph $G^-$, a node $u \in V^-$, and negative edge threshold $n$, the disgruntlement of the node $u$ is defined as 
\begin{align}\vspace{-0.3cm}
\begin{aligned}
 D(u) &= D^{self}(u)+D^{neib}(u) \\
 D^{self}(u) &= \begin{dcases*}
 0, & if $u$ is not a key node,\\
 |N(u, G^-)| - n+1, & if $u$ is a key node
 \end{dcases*} \\
 D^{neib}(u) &= \sum_{w\in N(u, G^-)} 1, \forall N(w, G^-)\geq n
\end{aligned}
\end{align}\vspace{-0.3cm}
\end{definition}

\begin{example}
In Figure~\ref{fig:example}, suppose that $n=2$. The disgruntlement of the following nodes can be computed.
For node $a$, node $a$ is not a key node ($D^{self}(a)=0$) and its neighbour has no key node ($D^{neib}(a)=0$); thus, $D(a)=0$. 
Similarly, $D(c)=0$. 
Node $b$ is known not to be a key node ($D^{self}(b)=0$), but its neighbour node $d$ is a key node; thus, $D^{neib}(b)=1$. 
Since node $d$ is a key node, we notice that $D^{self}(d)=3-2+1=2$ and $D^{neib}=0$; thus, $D(d)=2$. 
\end{example}

\subsubsection{Computing a lower-bound}\label{sec:lowerbound}

In this section, we present an approach to find a lower-bound of the followers. Let us recall the definition of the followers (Definition~\ref{def:follower}). When a node in the positive graph is removed, a union of the node and a set of nodes that is removed together due to the positive edge constraint are named as the followers. 
In this section, given a node $v\in V^+$, we discuss how to find a lower-bound, i.e., a lower-bound of the node $v$ named $LB(v)$ is always smaller than $|F(v)|$.
The VD-node (verge of death) is first defined to identify the lower-bound of the followers.
\begin{definition}
(\underline{VD-node}).
Given a graph $G$, and integer $p$, a node $v\in V$ is named a  VD-node if its degree and coreness\footnote{Given a graph $G$ and a node $v$, the coreness of the node $v$ is $q$ if the node $v$ belongs to the $q$-core but not to $(q+1)$-core.} is $p$. 
\end{definition}
The VD-node implies that the node will be deleted after any neighbour nodes of the VD-node are removed since the degree of the VD-nodes is exactly $p$.

\begin{definition}
(\underline{VD-cc}).
Given a graph $G$, and a set of VD-nodes, a set of connected components of VD-nodes are denoted as VD-ccs. 
\end{definition}

If any node in a VD-cc is removed, all the nodes in the connected components are removed together. Thus, finding the lower-bound of the followers can be achieved as follows.
First, all the VD-nodes ($O(|E^+|)$) are computed and VD-ccs ($O(|V|+|E^+|)$) are found. 
Next, given a node $v\in V$,  $N(v, G^+)$ can be computed. 
If any neighbour node $w\in N(v, G)$ is in VD-nodes, the VD-cc of the node $w$ can be obtained. After combining all the VD-ccs of the neighbour nodes, the minimum size of the followers of the node $v$ can be computed by summarising all the sizes of VD-ccs. 
A benefit of using VD-ccs is that to find the lower-bounds of all the nodes, repeatedly finding VD-ccs is not required. 

\begin{figure}[ht]
\centering
\includegraphics[width=0.99\linewidth]{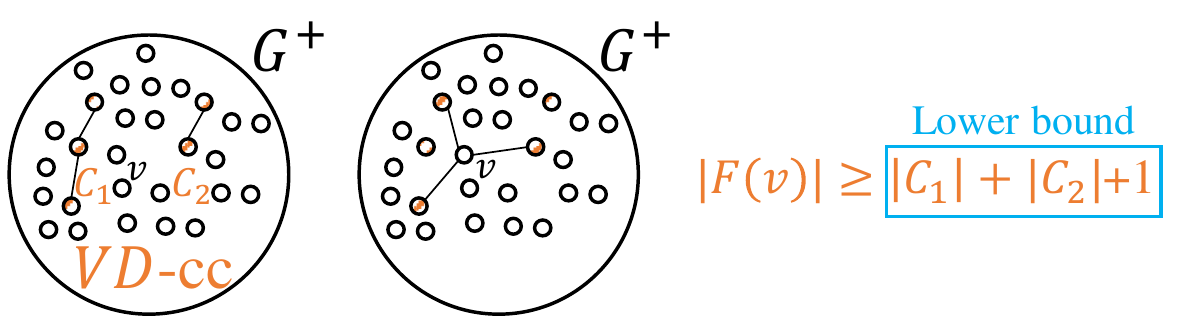}
\vspace{-0.1cm}
\caption{Computing a lower-bound}
\label{fig:lowerbound}
\end{figure}

\subsubsection{Computing an upper-bound}\label{sec:upperbound}

In this section, we present an approach to find an upper-bound of the number of followers by incorporating the hierarchical structure of the $k$-core~\cite{seidman1983network}. To describe the idea of identifying the upper-bound, we must first introduce some definitions. 

\begin{definition}
(\underline{CCNode}).
Given a positive graph $G^+$, and integer $k$, we define a set of connected subgraphs of induced subgraphs by the $k$-core as a set of CCNodes, i.e., each connected subgraph of $k$-core is a CCNode. 
\end{definition}

We can consider that any pair of two CCNodes is not connected and the number of maximal CCNodes is $\frac{|V^+|}{(k+1)}$. By utilising the CCNode, we define CCTree. 

\begin{definition}
(\underline{CCTree}).
CCTree is a tree consisting of a dummy root node and a set of CCNodes. CCNodes in the same tree level imply that they belong to the same $x$-core, and a parent and a children pair of the CCNodes which are connected in the CCTree implies that the parent CCNode (in $x'$-core) contains the children CCNode (in $(x'+1)$-core). 
\end{definition}

The height of the CCTree is the \underline{$c^{max}$ - $p+1$} where $c^{max}$ is the maximum coreness in the positive graph. Identifying the upper-bound is described as follows. 

\begin{strategy}
A CCTree is first constructed based on the input graph $G^+$ and positive edge threshold $p$ by making a dummy node and adding it to the CCTree as a root node. Next, $p$-core is computed and all the connected components induced by $p$-core are found. The connected subgraphs can be the CCNodes and they can be the children of the root node with level 1 of the CCTree. 
Then, for each CCNode $C_i$, we check whether $C_i$ contains $p'$-core with $p'=p+l$ where $l$ is the level of the $C_i$. For example, if a CCNode is in level 2, it indicates that the CCNode is a subgraph induced by $(p+1)$-core. Thus, we check whether $(p+2)$-core exists in the nodes of the CCNode. 
If $p'$-core exists, a new CCNode $C_j$ is added by computing the connected components of $p'$-core. Then, $C_i$ is made to be the parent of $C_j$ if $C_j \subseteq C_i$. This process is repeated.

After constructing a CCTree, the upper-bound given a node $v$ is ready to be calculated. First, a starting CCNode of the node $v$ is found. To identify a starting CCNode, we find a CCNode $cur$ which contains the node $v$ while the level of the $cur$ in the CCTree is the largest among the CCNodes which contain the node $v$. The \textit{children immutable size} is then computed. As the name implies, the size of the node sets is found which are not deleted after removing the node $v$ since they are already sufficiently connected (children immutable). 

\spara{Children immutable size $CI(.)$.} The children immutable size is the sum of the children's size. Since the children of $cur$ do not contain the node $v$, and their coreness value is larger than $cur$, i.e., the nodes in the children of the $cur$ are already sufficiently connected, they will not be removed after removing the node $v$. 

\begin{figure*}[t]
\centering
\includegraphics[width=0.99\linewidth]{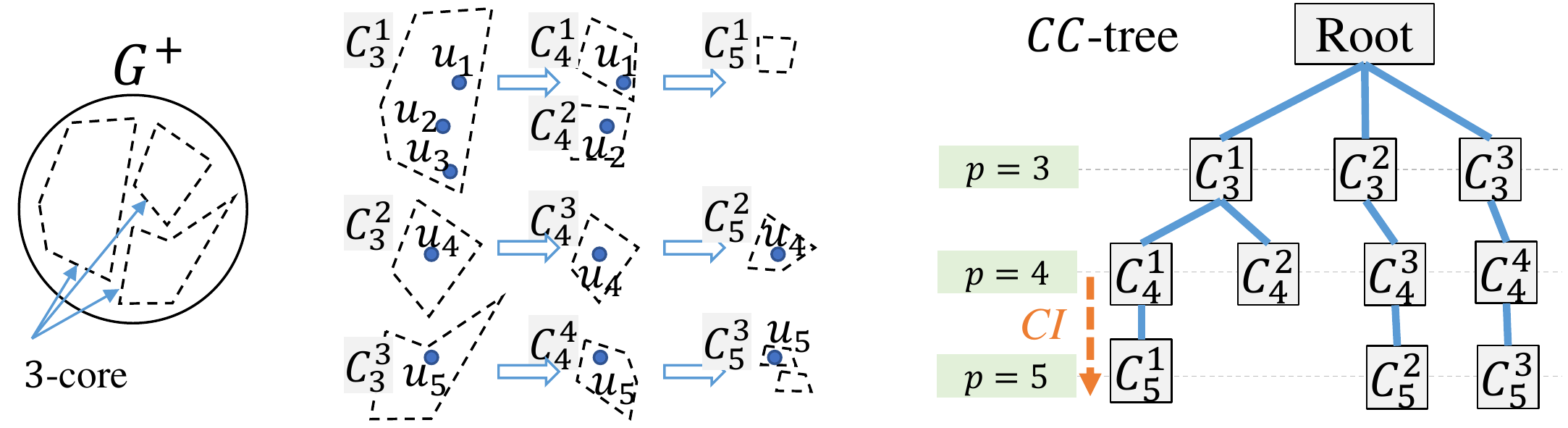}
\vspace{-0.1cm}
\caption{Computing an upper-bound}
\label{fig:upperbound}
\end{figure*}
By negating $CI(.)$ from $|cur|$, we can find upper bound. 
\begin{align}
 UB(u) = |cur| - CI(u)
\end{align}
\end{strategy}

\begin{example}
Figure~\ref{fig:upperbound} shows a CCTree and CCNodes when $p=3$. We assume that the input graph is $G^+$. Let find an upper-bound of the followers of the node $u_1$. 
We notice that the starting CCNode is $C_4^1$ since $C_4^1$ is deeply located than $C_3^1$. We then compute $CI(u_1)$ by summarising the children's size. Hence, $CI(u_1)=$ $|C_5^1|$. 
Therefore, the upper-bound of the node $u_1$ is $|C_4^1|-|C_5^1|$. 
\end{example}

\begin{algorithm}[t]
\SetKw{ret}{return}
\SetKwData{D}{disgruntlement}
\SetKwData{T}{T}
\SetKwFunction{computedisgruntlement}{computeD}
\SetKwFunction{computingVDcc}{computingVDcc}
\SetKwFunction{computecoreness}{computeCoreness}
\SetKwFunction{remove}{remove}
\SetKwFunction{computeCandidates}{computeCandidates}
\SetKwFunction{constructCoreTree}{constructCoreTree}
\SetKwFunction{neighbours}{neighbours}
\SetKwFunction{updatecoreness}{updateCoreness}
\SetKwFunction{updateCT}{updateCT}
\SetKwFunction{computeLowerBound}{computeLowerBound}
\SetKwFunction{computeUpperBound}{computeUpperBound}
\SetKwInOut{Input}{input}
\SetKwInOut{Output}{output}
\Input{Signed graph $G$, and positive integers $p$ and $n$}
\Output{$(p,n)$-core}
\tcp{same with lines 1-3 in Algorithm~\ref{alg:FBA}}
$c[v] \leftarrow $\computecoreness{$H^+$}, $\forall v\in V_H$\;
$ct \leftarrow$ \constructCoreTree{}\;
\While{$\gamma(H^-) < n$}{
$D(v) \leftarrow $ \computedisgruntlement{$v, H^-$}, $\forall v\in T$ and $D(v)\neq 0$\;
 $M \leftarrow $ \computingVDcc{$H^+, c[]$}\;
 $LB^* \leftarrow $ \computeLowerBound{$v, ct$} \;
 $UB^* \leftarrow $ \computeUpperBound{$v, M$}\;
 $Cand \leftarrow$ \computeCandidates{$LB^*, UB^*$}\;
 $u \leftarrow $ $\argmax_{v\in Cand} O(v) $\;
 $F(u) \leftarrow$ \computeFollower{$u, H^+$}\;
 $c.$\remove{$c$, $F(u)$}\;
 $H^+ \leftarrow$ \removeNode{$H^+, F(u)$}\; 
 $H^- \leftarrow$ \removeNode{$H^-, F(u)$}\; 
 $T \leftarrow T \setminus F(u)$\;
 \updatecoreness{$c$} \tcc{apply Purecore~\cite{sariyuce2013streaming}}
 \updateCT{$ct$, $F(u)$}\;
}
\ret $V_H$\;
\caption{Disgruntled follower-based algorithm(\DFBA)}
\label{alg:DFBA}
\end{algorithm}

\subsubsection{Algorithmic procedure}\label{sec:DFBA_algorithmic}
We designed the {\DFBA} by combining the lower-bound and upper-bound, and disgruntlement. Here, we first discuss our strategy and the objective function of {\DFBA}.

\begin{strategy}
Instead of computing all the nodes' followers, we aim to compute only a few followers for efficiency. We first compute the following three values : 
(1) Disgruntlement score $D(.)$; 
(2) lower-bound $LB(.)$; and 
(3) upper-bound $UB(.)$. 
By utilising disgruntlement and the $LB(.)$ and $UB(.)$, we can set the new lower-bound and upper-bound. Note that the following equation always holds. 
\begin{align}\label{eq:inequality}
\vspace{-0.2cm}
\begin{aligned}
 LB(.) \leq |F(.)| \leq &UB(.) 
 \Rightarrow\frac{D(.)}{LB(.)} \geq \frac{D(.)}{|F(.)|} \geq \frac{D(.)}{UB(.)} \\
 \Rightarrow & UB^*(.) \geq O(.) \geq LB^*(.) 
\end{aligned}
\vspace{-0.2cm}
\end{align}
In our algorithm, we utilise Equation~\ref{eq:inequality} in the pruning strategy. Given two nodes $u$ and $v$, if $LB^*(v) > UB^*(u)$, computing the followers of the node $u$ is not necessary. This is because the upper-bound is smaller than the lower-bound of the node $v$. Hence, we first compute the node having the largest $LB^*(.)$ and then find a set of candidate nodes to compute the followers. The detailed procedure of {\DFBA} is described in Algorithm~\ref{alg:DFBA}. 
\end{strategy}

\spara{Complexity.} 
The time complexity of {\DFBA} is $O(|V|^2(|V|+|E^+|)+|V||E^-|)$ since in the worst case,  all the nodes' followers in $|V|$ iterations must be computed. Thus, it takes $O(|V|^2(|V|+|E^+|)$ and computing disgruntlement for every iteration takes $O(|V||E^-|)$. 

\subsection{Fast-Circle algorithm (\FCA)}\label{sec:FCA}

\subsubsection{Motivation}

In Sections~\ref{sec:FBA} and \ref{sec:DFBA}, we compute the exact followers to find the best node to be removed. 
In this section, we propose a new approach by not computing the exact number of followers. To estimate the size of the followers, the characteristics of the followers in a signed network are verified. A simple observation is the relationship between the degree of the nodes and the size of followers, which is presented in Figure~\ref{fig:1hop}. When the regression line (red-coloured) to fit the data distribution is plotted, we observe that the degree and the size of followers are correlated. 
Thus, selecting a node having a large degree may have a large number of followers. Figure~\ref{fig:2hop} shows the correlation between the two hop degree of nodes and the size of followers, which is observed to still have a correlation. This observation indicates that when a node having a small degree is selected, it may have few followers. The challenge of this approach is in \textit{efficiently compute the $r$-hop degree} To handle the challenge, we utilise the HyperANF technique.

\begin{figure}[ht]
\centering
\includegraphics[width=0.99\linewidth]{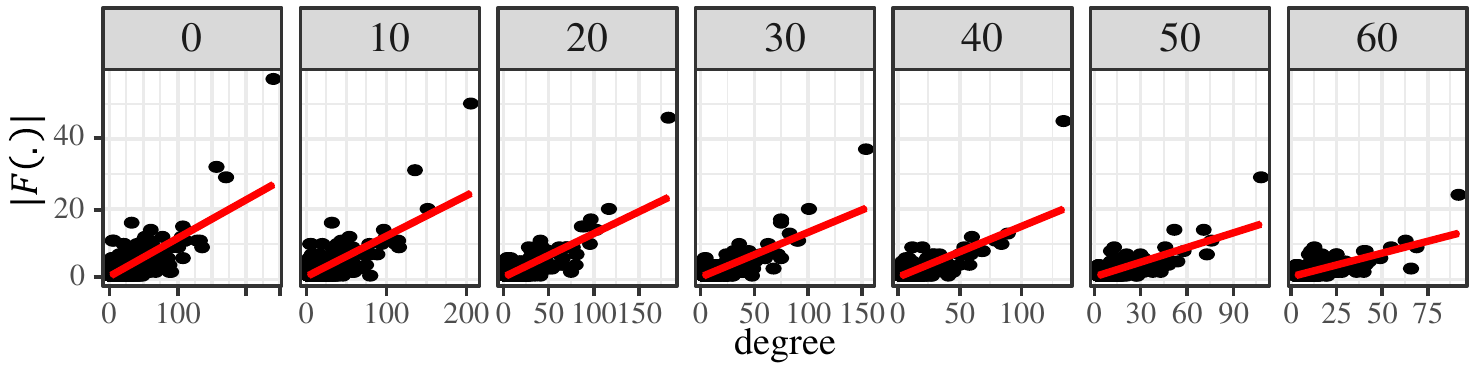}
\vspace{-0.3cm}
\caption{Node degree and the size of followers}
\label{fig:1hop}
\end{figure}

\begin{figure}[ht]
\centering
\includegraphics[width=0.99\linewidth]{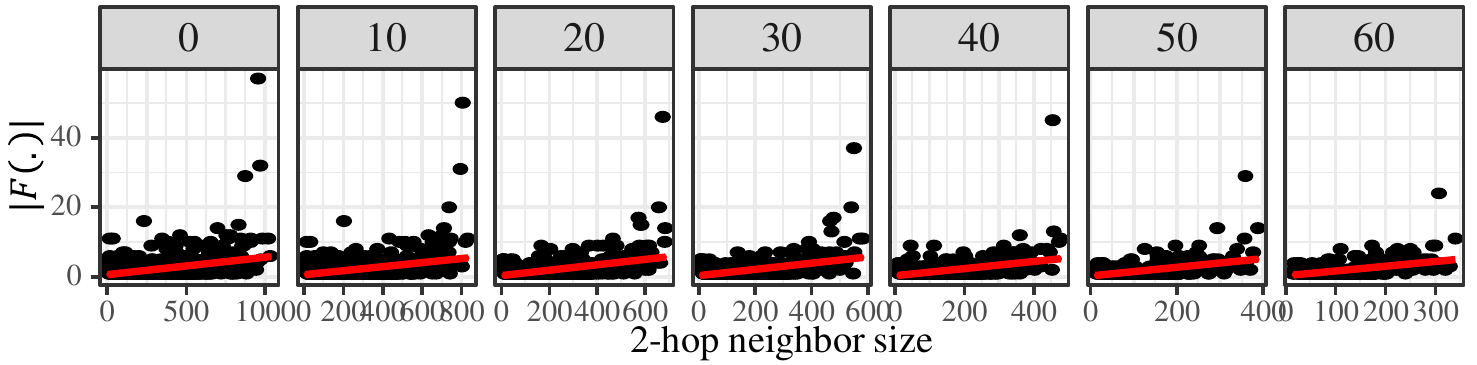}
\vspace{-0.3cm}
\caption{Two-hop degree and the size of followers}
\label{fig:2hop}
\end{figure}

\subsubsection{HyperANF}
HyperANF~\cite{boldi2011hyperanf} is a state-of-the-art algorithm for computing the approximated neighbourhood cardinality of a network. 
It utilises the HyperLogLog counter~\cite{flajolet2007hyperloglog} which is a statistical counter requiring $O(\log \log {n})$ bits. 
Therefore, HyperANF can compute the approximated number of reachable nodes within a specific distance from a node.
Unfortunately, the time complexity of the HyperANF is not revealed; however, its running time is expected to be approximately $O(r(|V|+|E|))$, where $r$ denotes a maximum distance since HyperANF is a kind of extension of ANF~\cite{palmer2002anf} requiring $O((|V|+|E|)h)$. 

\subsubsection{Algorithmic procedure}

\begin{strategy}
The high-level idea of {\FCA} is to use the similar objective function in Equation~\ref{eq:main_obj} to select the best node to be removed by not computing the exact value $F(.)$ to find the best node,  the estimated number of followers $\tilde{F}(.)$ is computed. 
To estimate the number of followers, we utilise HyperANF~\cite{boldi2011hyperanf}. 
After selecting the best node to be removed, the node is selected and the exact followers is determined. 
Next, $\tilde{F}(.)$ values are updated by computing the shortest paths within a specific distance $r$. 
\end{strategy}

A pseudo description of {\FCA} can be checked in Algorithm~\ref{alg:FCA}

\begin{algorithm}[t]
\SetKw{ret}{return}
\SetKwData{T}{T}
\SetKwFunction{computeUpperBound}{computeUpperBound}
\SetKwFunction{pcore}{pcore}
\SetKwFunction{follower}{follower}
\SetKwFunction{getBest}{getBest}
\SetKwInOut{Input}{input}
\SetKwInOut{Output}{output}
\Input{Signed graph $G$, and positive integers $p$, $n$, and $r$}
\Output{$(p,n)$-core}
$H\leftarrow$ $G[$\pcore{$G$}$]$\;
$\forall v\in H$, compute $D(v)$\;
$\forall v\in H$, compute $\tilde{F}(v)$\;
$\forall v\in H$, compute $O^*(v)=\frac{D(v)}{|\tilde{F}(v)|}$\;
\While{$\gamma(H^-) < n$}{
 $u \leftarrow$ \getBest($O^*$)\;
 $F(u) \leftarrow $ \follower{$u$, $H^+$}\;
 update $O^*$ and $\tilde{F}(.)$\;
 $H^+ \leftarrow$ \removeNode{$H^+, F(u)$}\; 
 $H^- \leftarrow$ \removeNode{$H^-, F(u)$}\; 
}
\ret $V_H$\;
\caption{Fast Circle Algorithm (\FCA)}
\label{alg:FCA}
\end{algorithm} 

\spara{Time complexity.} Let denote $H$ as a time complexity for HyperANF. Then, the time complexity of our {\FCA} is as follows. 
\begin{itemize}
 \item Maximum number of iterations is $O(|V|)$.
 \item Computing HyperANF takes $O(H)$.
 \item Computing followers takes $O(|V|+|E|)$.
 \item Dijkstra's shortest path algorithm takes $O(|V|+|E|\log{|V|})$. 
\end{itemize}

Hence, the time complexity of {\FCA} algorithm is $O(|V|(|V|+|E|\log{|V|}+|V|+|E|) + H)$.

\section{Experiments}\label{sec:experiments}
In this section, we evaluate our algorithms using eight real-world networks and three synthetic networks. All the experiments were conducted on Ubuntu 18.04 with 128GB memory and a 2.50GHz Xeon CPU E5-4627 v4. For implementation, JgraphT library~\cite{jgrapht} was used. 

\spara{Dataset.}
Table~\ref{tab:real_dataset} reports the statistics of the eight real-world networks. In our work, we do not consider the weighted temporal networks; thus, we convert Alpha and OTC datasets~\cite{kumar2016edge} as signed networks by ignoring the weights, and keeping the most recent edges with its sign. 
All the datasets are publicly available. 
In Table~\ref{tab:real_dataset},  $|V|$ is the number of nodes, $|+|$ (or $|-|$) is the number of positive (or  negative) edges, MC is the maximum positive coreness value of a network, and \# $\Delta$ is the number of triangles in the positive graph. 

\begin{table}[ht]
\caption{Real-world datasets}
\label{tab:real_dataset}
\begin{tabular}{c|c|c|c|c|c}
\hline
   Dataset                & $|V|$       & $|+|$       & $|-|$       & MC & \# $\Delta$               \\ \hline \hline
 Alpha~\cite{kumar2016edge}                  & 3,783     & 12,759    & 1,365     & 18            & 16,838                         \\  
                                     OTC~\cite{kumar2016edge}                    & 5,881     & 18,250    & 3,242     & 19             &  23,019                         \\ 
                                     Epinions(EP)~\cite{leskovec2010signed}           & 131,828	      & 590,466 &   120,744   & 120           & 3,960,165                      \\ 
                                     SD0211~\cite{leskovec2010signed} & 82,140    & 382,167   & 118,314   & 54            & 418,832                        \\ 
                                     SD0216~\cite{leskovec2010signed} & 81,867    & 380,078   & 117,594   & 54            & 414,903                        \\ 
                                     SD1106~\cite{leskovec2010signed} & 77,350    & 354,073   & 114,481   & 53            & 395,289                        \\ 
                                     Wiki-E(WE)~\cite{brandes2009network}           & 116,836   & 774,785   & 1,253,086 & 88            & 4,450,048                      \\ 
                                     Wiki-I(WI)~\cite{maniu2011casting}   & 138,587   & 629,523   & 86,360    & 53            & 2,599,698                      \\ \hline
 \hline
\end{tabular}
\end{table}

Since real-world networks are relatively small-sized signed networks, we synthetically generate signed networks by utilising existing non-signed networks~\cite{yang2015defining}. The statistics can be checked in Table~\ref{tab:synthetic}. 

\begin{table}[ht]
\caption{Synthetic datasets}
\label{tab:synthetic}
\centering
\begin{tabular}{c|c|c|c|c|c}
\hline
   Dataset                & $|V|$       & $|+|$       & $|-|$       & MC & \# $\Delta$               \\ \hline \hline
   
Amazon~\cite{yang2015defining} & 334,863  & 925,872    & 1,851,744     & 6    &  667,108                         \\ 
DBLP~\cite{yang2015defining} & 317,080  & 1,049,866    & 2,099,732     & 113    &  2,224,650                         \\ 

Youtube~\cite{yang2015defining} & 1,134,890  & 2,987,624    & 5,975,248     & 51    &  3,056,379                         \\ 
 \hline \hline
\end{tabular}
\end{table}

\spara{\bf Algorithms.} To the best of our knowledge, our $(p,n)$-core does not have any direct competitor in  previous literature since all the previous work focuses on the different objective function with different problems. Hence, directly comparing the size of the resultant subgraph is unfair. 
Thus, we report the result of our proposed algorithms in these experiments. Since both {\FBA} and {\DFBA} are not scalable due to their time complexity, we only report the results of both algorithms if the algorithms were terminated within 24 hours. {\FBA} is only applicable in OTC and Alpha datasets, and {\DFBA} cannot be applicable in WE.

\begin{figure}[ht]
\centering
\includegraphics[width=0.99\linewidth]{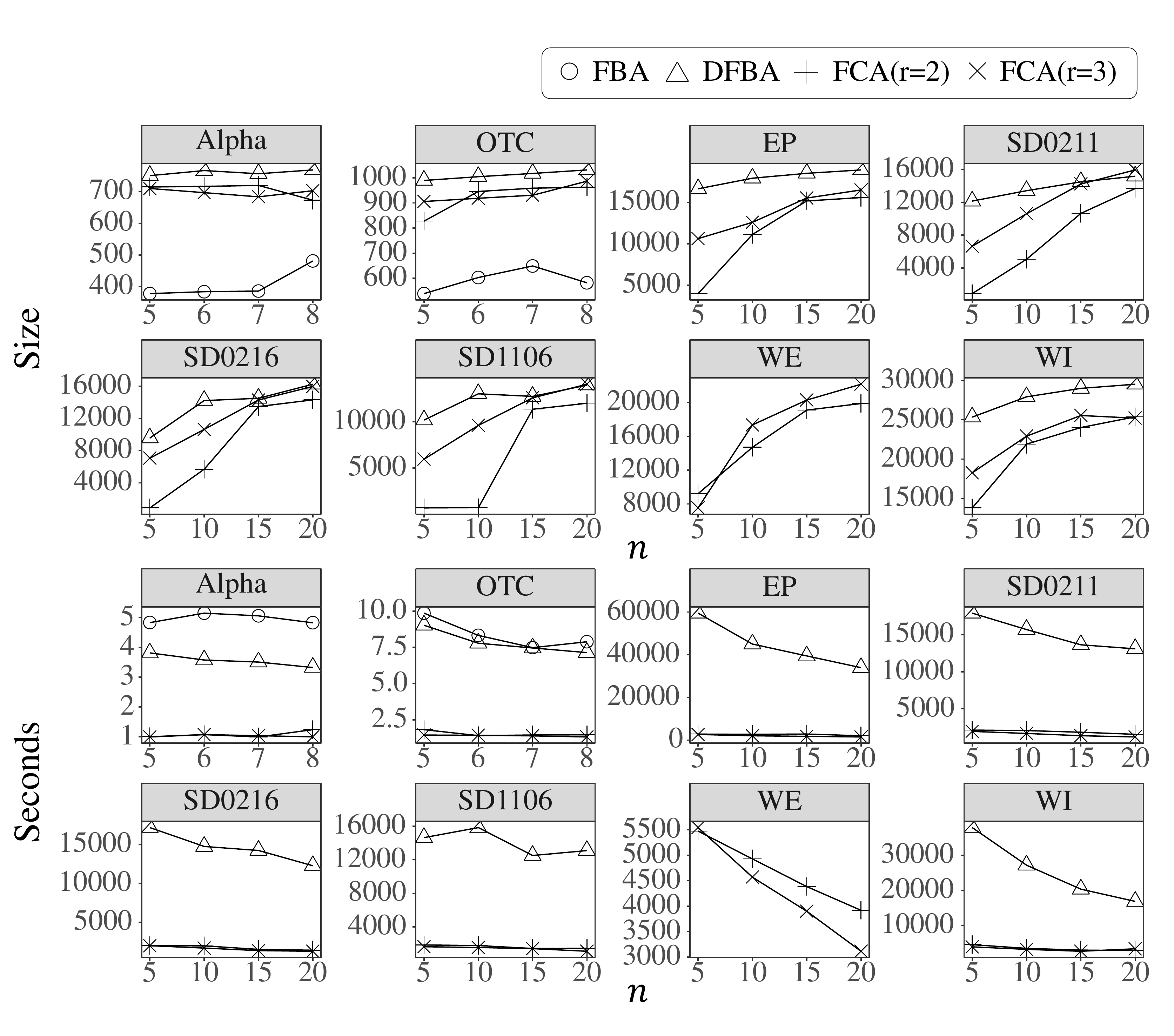}
\vspace{-0.2cm}
        \caption{Varying the variable $n$}
        \label{fig:var_n} 
\end{figure}

\spara{Effect on $n$.} 
In Figure~\ref{fig:var_n}, we fix the parameter $p=5$ then change the negative edge threshold $n$ from $5$ to $20$ (for Alpha and OTC datasets, $n$ is varied from $5$ to $8$ since the size of both networks is small). 
When the negative edge threshold $n$ becomes large, our algorithms return larger $(p,n)$-core since the large $n$ indicates that we allow more negative edges in $(p,n)$-core. In addition, we observe that when $n$ becomes large, the algorithm will be finished early since we do not need to remove many nodes in the removing procedure.  

For all the cases, we observed that {\DFBA} returns the best result. We verified that {\FBA} returns small-sized $(p,n)$-core since it does not consider the negative edges when it removes a node. We also verified that both {\FCA}s return reasonable solutions, but their effectiveness is limited compared with {\DFBA} since it uses the approximated size of the followers. We observe that {\FCA}(r=3) normally returns better results compared with {\FCA}(r=2) since it considers much more structural information.

\begin{figure}[ht]
\centering
\includegraphics[width=0.99\linewidth]{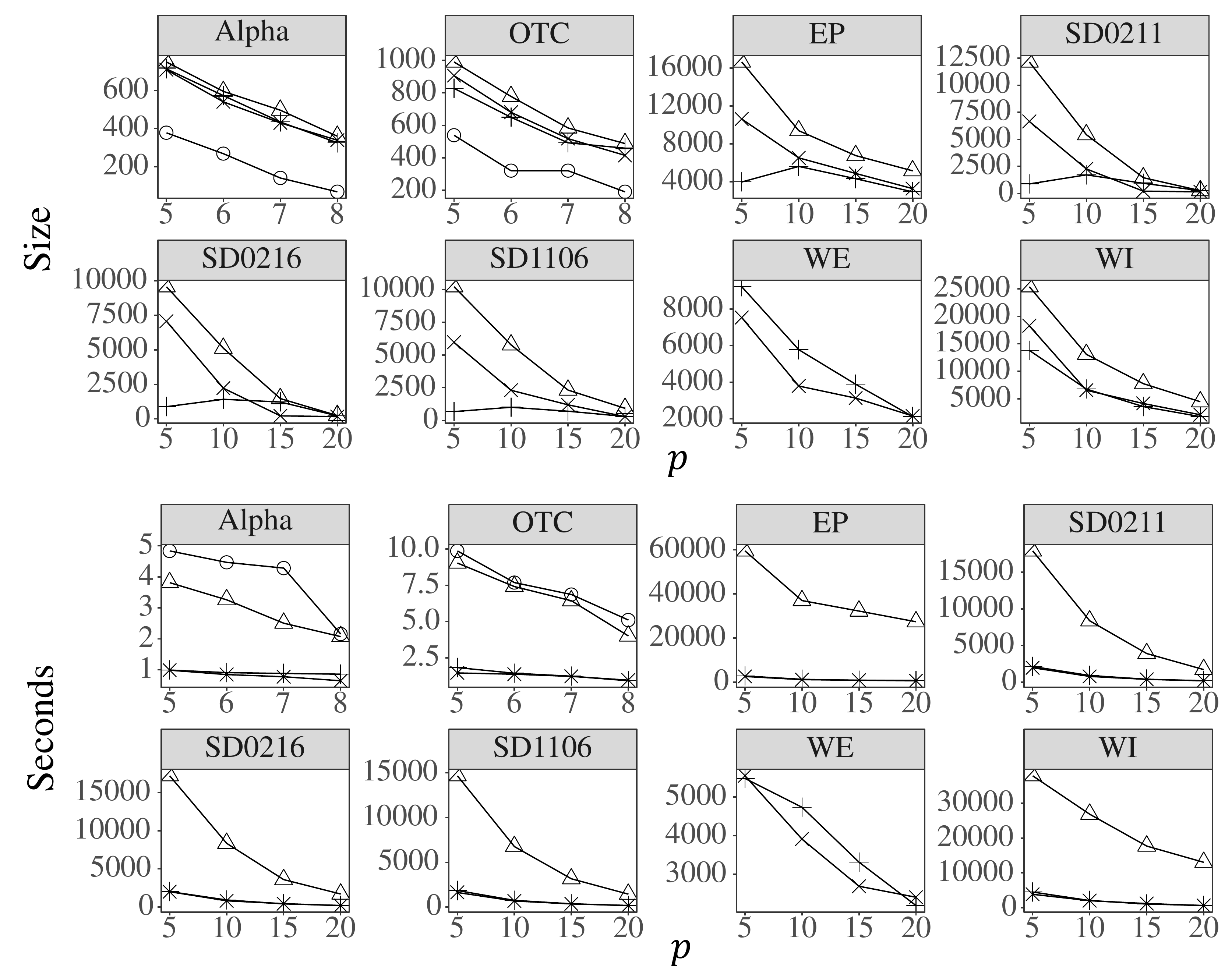}
\vspace{-0.2cm}
        \caption{Varying the variable $p$ (same legend with Figure~\ref{fig:var_n})}
        \label{fig:var_p} 
\end{figure}

\spara{Effect on $p$.} In Figure~\ref{fig:var_p}, the parameter $n$ is fixed as $5$, and then the positive edge threshold $p$ is varied from $5$ to $20$ to observe the trends (we change $n$ from $5$ to $8$ for Alpha and OTC datasets). We identified that when the $p$ value becomes large, the proposed three algorithms consistently return small-sized $(p,n)$-core as a result since the larger $p$ implies that the resultant subgraphs are more cohesive through the positive edges. We also observed that when $p$ becomes large, the algorithms are terminated early since $p$-core returns many small-sized results, i.e., the nodes to be removed are comparably limited. 

These results followed the same trends as Figure~\ref{fig:var_n}. We noticed that {\DFBA} returns the best result and {\FCA}(r=3) returns comparable results, and the {\FCA} algorithms with different radius parameters are much faster than other algorithms since it does not compute the exact number of followers for every iteration.

\begin{figure}[ht]
\centering
\includegraphics[width=0.99\linewidth]{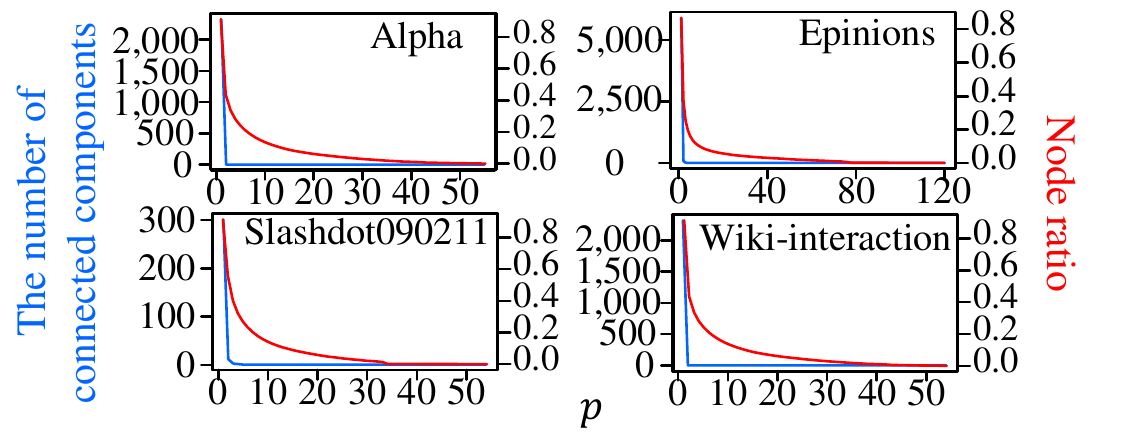}
\vspace{-0.2cm}
\caption{\# of connected components and the node ratio}
\label{fig:cc}
\end{figure}

\spara{Analysis of the pruning strategy.}
Figure~\ref{fig:cc} shows the number of connected components (blue-coloured) and the ratio of the nodes (red-coloured) when  the parameter $p$ is changed from $1$ to the maximum coreness value in real-world networks. We observed that there is a single large connected component containing all the nodes when $p$ is larger than or equal to $3$. We also see that the difference between any adjacent $p$ values is0 not significant. That implies that the children immutable size (CI) may be helpful to find the proper upper-bound to improve  efficiency.  

\begin{figure}[ht]
\centering
\includegraphics[width=0.95\linewidth]{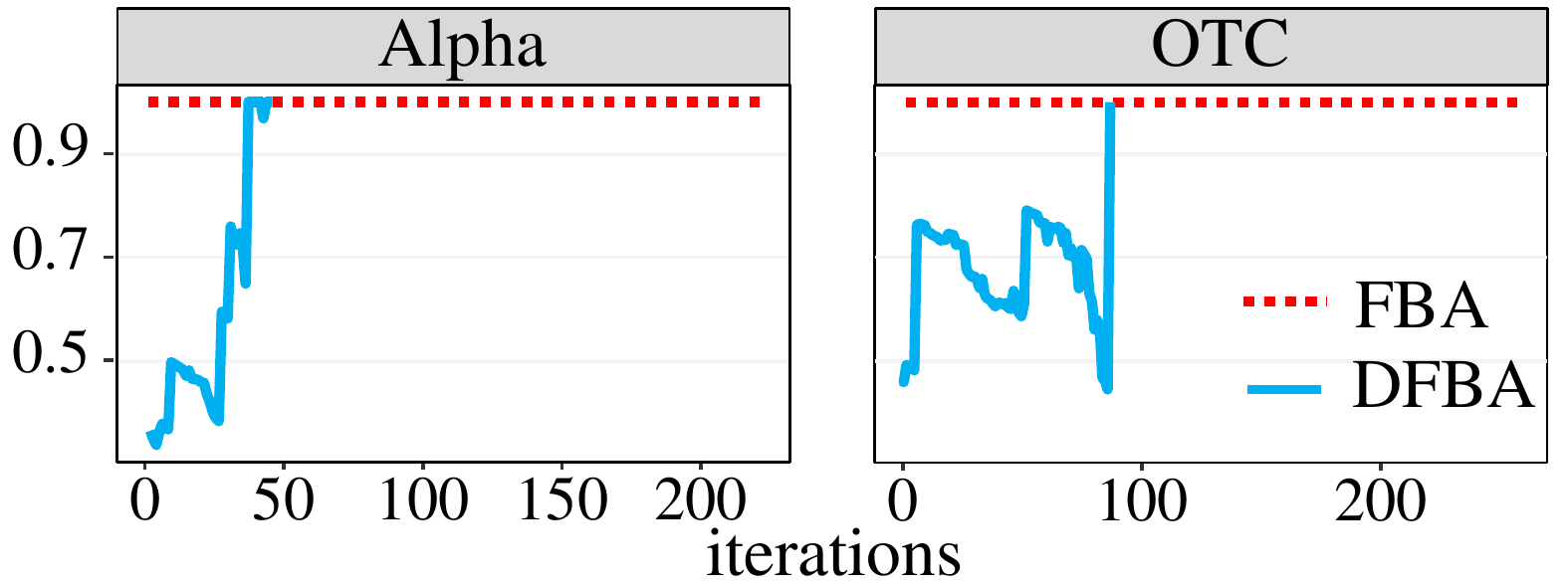}
\vspace{-0.4cm}
\caption{Ratio of computing followers in Alpha and OTC networks ($p=5, n=5$)}
\vspace{-0.4cm}
\label{fig:pruning}
\end{figure}

\spara{Effect on pruning strategy.}
In this section, we verified that our proposed {\DFBA} approach can avoid computing followers in Alpha and OTC networks. Figure~\ref{fig:pruning} reports the ratio of the nodes to compute followers among all the remaining nodes in {\DFBA} and {\FBA}. The {\DFBA} is verified to have less iterations than {\FBA}. More specifically, in Alpha and OTC networks, {\DFBA} has less than $30\%$ iterations compared with the number of iterations of {\FBA}. 
We also observed that {\DFBA} does not compute many followers compared with {\FBA} since the pruning strategy can affect the avoidance of computing all the followers.

\begin{figure}[ht]
\centering
\includegraphics[width=0.99\linewidth]{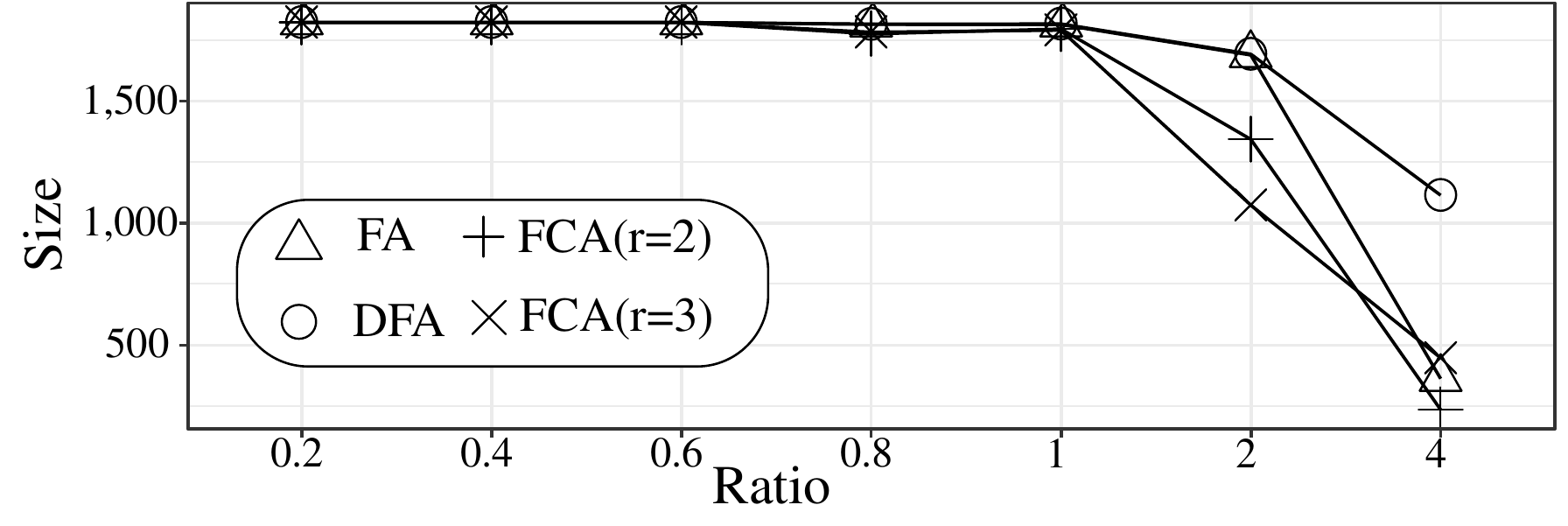}
\vspace{-0.3cm}
\caption{Varying the ratio}
\vspace{-0.4cm}
\label{fig:ratio_size}
\end{figure}
\spara{Effect on negative edge ratio.} By using the LFR benchmark dataset~\cite{lancichinetti2008benchmark}, we varied the number of injected negative edges. The value $\alpha=\frac{|-|}{|+|}$ is varied from $0.2$ to $4.0$ where $|+|$ is the number of positive edges and $|-|$ is the number of negative edges. We observed that when the ratio is smaller than $1$, it has no significant difference. However, when the values are larger than or equal to $2$, the size of the identified solution becomes decrease significantly. When $\alpha=4$, the {\FBA} contains only $19.7\%$ of the nodes compared with $\alpha=0.2$.

\begin{figure}[ht]
\centering
\includegraphics[width=0.99\linewidth]{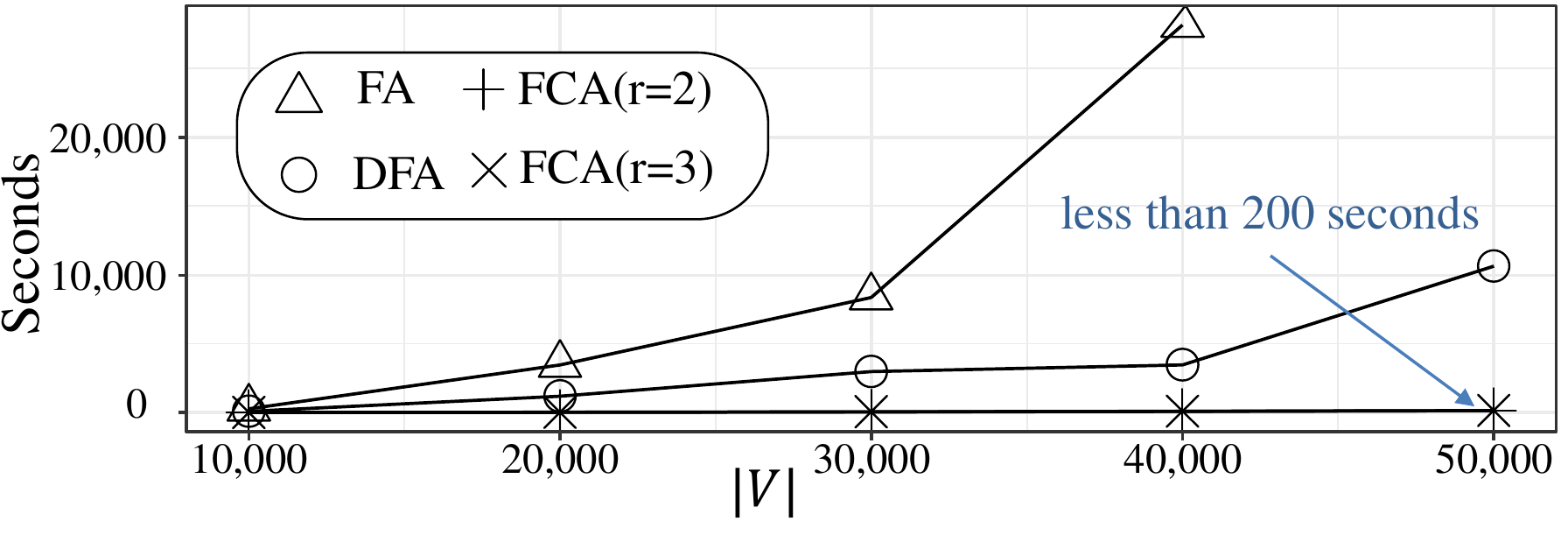}
\vspace{-0.1cm}
\caption{Scalability test ($p=3, n=3$)}
\vspace{-0.4cm}
\label{fig:scalability}
\end{figure}

\spara{Scalability test.} We used the LFR benchmark datasets~\cite{lancichinetti2008benchmark} to check the scalability of our proposed algorithms by varying the size of the nodes from $10K$ to $50K$. 
In the synthetic networks, we set the average degree as $5$, maximum degree as $50$, minimum community size as $10$, maximum community size as $500$, and community mixing parameter as $0.1$. We injected $\frac{|+|}{2}$ negative edges. 
We found that the proposed {\FBA} is slower than {\DFBA} since it must compute all the followers in every iteration. Even if the time complexity of {\DFBA} is the same as {\FBA}, due to the pruning strategy, we found that {\DFBA} is much faster than {\FBA}. We can verify that {\FCA} is over $60$ times faster than {\DFBA} when $|V|=50,000$ since it does not require computing the exact number of followers to find the best node to be removed.

\begin{figure}[ht]
\centering
\includegraphics[width=0.95\linewidth]{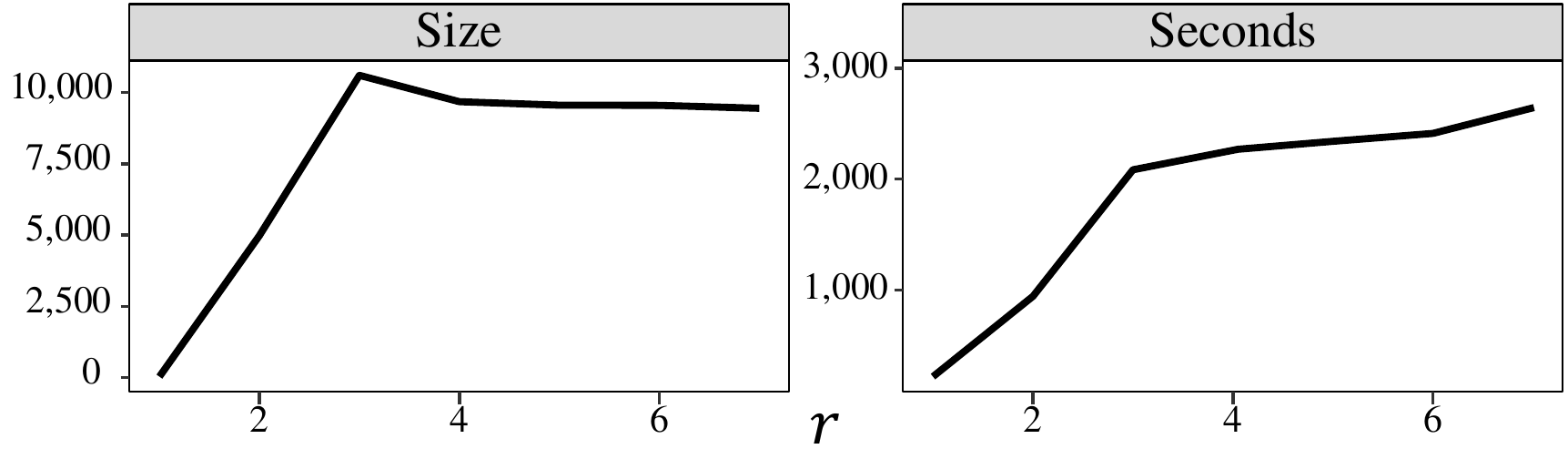}
\vspace{-0.4cm}
\caption{Varying radius $r$}
\vspace{-0.4cm}
\label{fig:radius}
\end{figure}

\spara{Varying the radius $r$ of {\FCA}.} In {\FCA}, the radius $r$ is a required user parameter. We check the result of $(p,n)$-core by varying the radius $r$ in the EP dataset. When $r$ becomes large, the running time is increased since the time complexity of HyperANF depends on the number of hops. One interesting point is the size of the solution. We observe that when $r=3$, the resultant subgraph returns the largest subgraph as a result. This implies that checking too much information may not be useful to estimate the number of followers. Therefore, in this paper, we set $r=2$ (for efficiency-focus mode) or $r=3$ (for effectiveness-focus mode).

\begin{figure}[ht]
\centering
\includegraphics[width=0.99\linewidth]{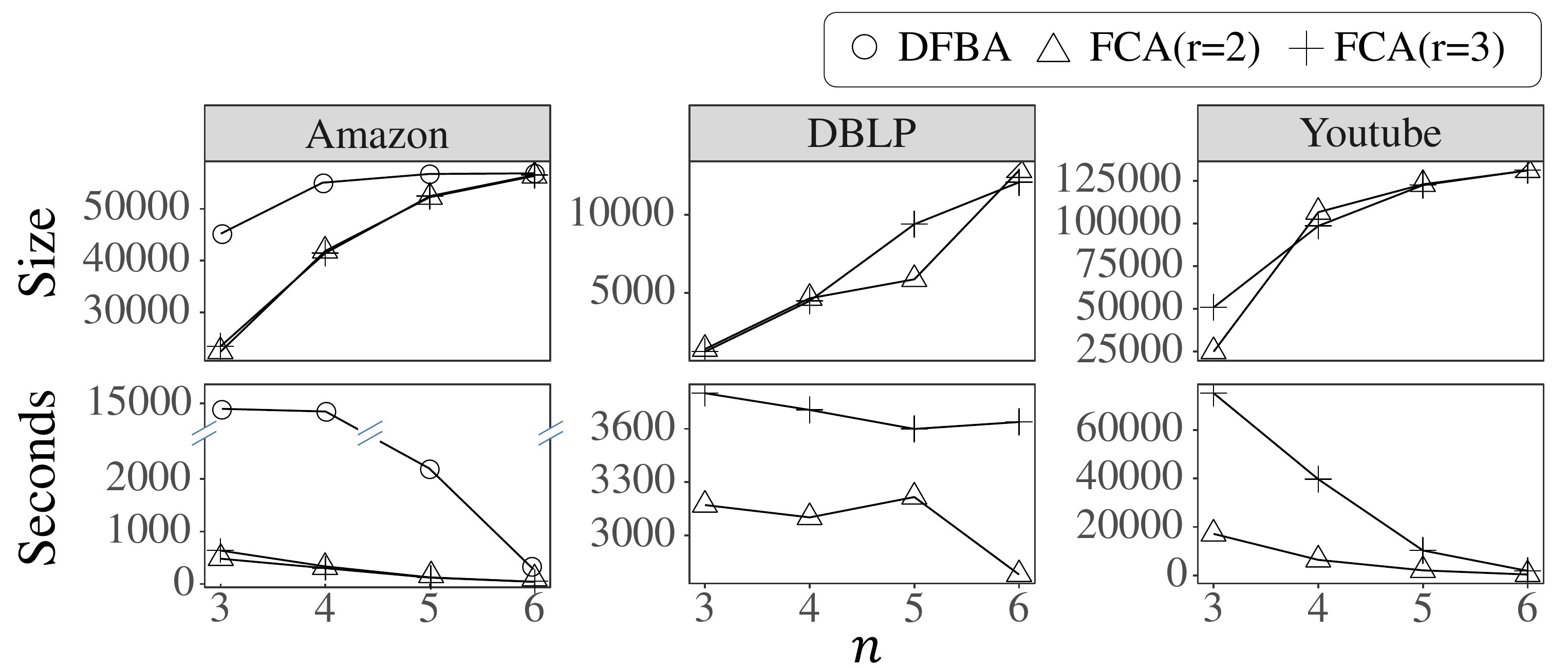}
\vspace{-0.2cm}
        \caption{Varying the variable $p$}
        \label{fig:syn_var_p} 
\end{figure}

\begin{figure}[ht]
\centering
\includegraphics[width=0.99\linewidth]{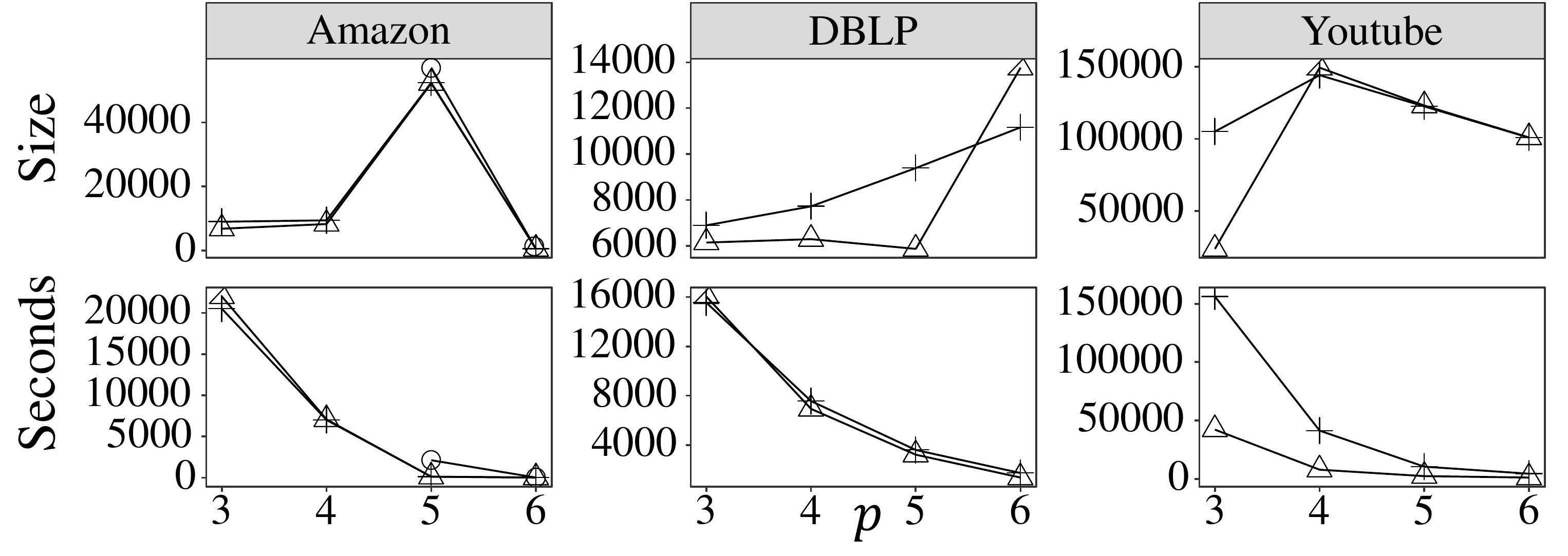}
\vspace{-0.2cm}
        \caption{Varying the variable $n$}
        \label{fig:syn_var_n} 
\end{figure}

\spara{Synthetic network.} 
In Figures~\ref{fig:syn_var_p} and \ref{fig:syn_var_n}, The parameter $n$ is fixed as $5$, and then the positive edge threshold $p$ varies from $3$ to $6$ (or $n$ from $3$ to $6$) to observe the trends in synthetic networks. We found that the trend of the result is the same as the experimental result of real-world networks. When the $p$ value becomes large, the resultant subgraph becomes small due to the stricter constraint, and when the $n$ value becomes large, the resultant subgraph becomes large due to the relaxed constraint. 
We noticed that a larger $r$ makes the algorithm slower since it must check additional structural information. Note that in some cases shown in Figure~\ref{fig:syn_var_n}, sometimes the resultant subgraph size becomes large when $p$ becomes large. This is because removing several nodes is beneficial in finding large-sized subgraphs since the negative edges are injected randomly ($p=5$ in Amazon, $p=6$ in DBLP, and $p=4$ in Youtube). However, if $p$-core returns a very small-sized solution, it may return a very small-sized result ($p=6$ in Amazon).

\begin{figure}[ht]
\centering
\includegraphics[width=0.95\linewidth]{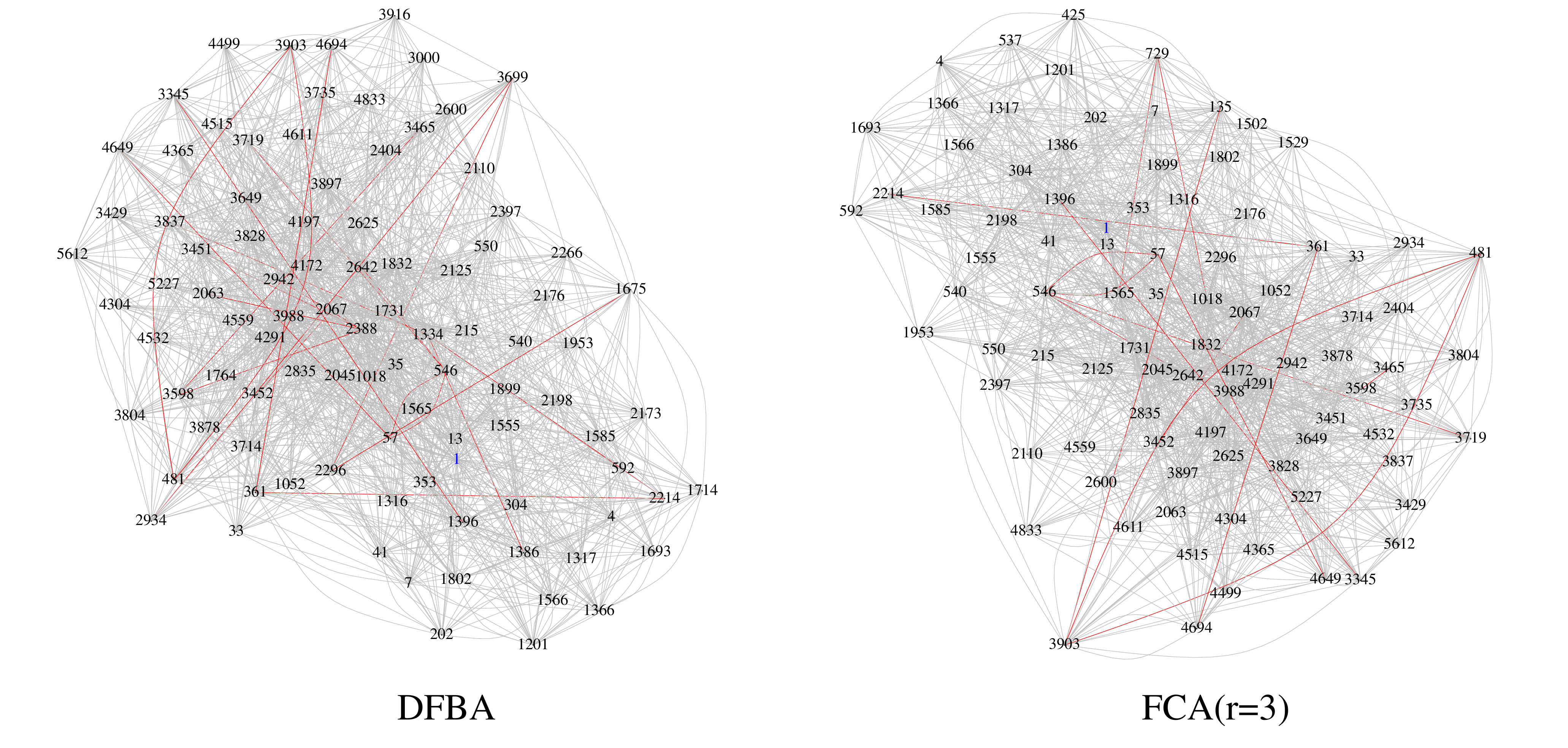}
\vspace{-0.2cm}
\caption{Identifying a community}
\label{fig:comm}
\end{figure}
\spara{Case study : Community Search.}
One famous community search model is to maximise the minimum degree~\cite{sozio2010community}. By setting the query node, we identified a community by fixing $n=5$ in the OTC dataset by utilising our proposed {\DFBA} and {\FCA} with $r=3$. Our query node is $v_1$. In Figure~\ref{fig:comm}, positive edges are grey-coloured, and the negative edges are red-coloured. The query nodes are densely connected to the other nodes in the community and there are few negative edges, and the degree of negative edges of the nodes is less than $5$. The statistics of the identified community are as follows. 

\begin{table}[ht]
\caption{Statistics of the identified communities}
\label{tab:cs_result}
\centering
\begin{tabular}{c|c|c|c|c}
\hline
     & $|V|$ & $|E|$   & Cluster Coefficient & \# of triangles \\ \hline \hline
{\DFBA} & 95  & 1,565 & 0.4099124           & 5,693           \\ \hline
{\FCA}(r=3)  & 91  & 1,430 & 0.4090217           & 4,833           \\ \hline \hline
\end{tabular}
\end{table}

\section{Related work}\label{sec:relatedwork}

In this section, we discuss three representative signed cohesive subgraph discovery problems and compare our cohesive subgraph model with them. 

First, In \cite{wu2020maximum}, Wu et al. study a signed ($k,r$)-truss problem.
They present balanced and unbalanced triangles to model ($k,r$)-truss. Specifically, given a signed network $G$ and two positive integers $k$ and $r$, a signed $(k,r)$-truss is a subgraph $S$ of $G$ which satisfies (1) $sup^+(e, S) \geq k$; (2) $sup^-(e, S) \leq r$; and (3) maximality constraint. Support $sup^+(e, S)$ (or $sup^-(e, S)$) indicates that the number of balanced (or unbalanced) triangles contain the edge $e$ in $S$. Wu et al. defines that a triangle is balanced if it contains an odd number of positive edges; otherwise, the triangle is unbalanced. 
In this paper, they show that the proposed $(k,r)$-truss is NP-hard to find the exact solutions and show that the different edge deletion orders may lead to different $(k,r)$-trusses. 
To find a solution, Wu et al. propose three algorithms: (1) the trivial approach which removes the edges based on the edge id; (2) the greedy edge removing approach which iteratively chooses an edge with the least followers; and (3) the triangle-based approach, which is a revised greedy algorithm by selecting the best edge which can remove many unbalanced triangles. Since this problem focuses on balanced and unbalanced triangles, it is different from our problem. 

Next, Giatsidis et al.~\cite{giatsidis2014quantifying} study the signed ($l^t, k^s$)-core problem in signed directed networks. Specifically, given a signed directed network $G$, and two parameters $k$ and $l$, it aims to find ($l^t, k^s$)-core which is a maximal subgraph $H$ of $G$ of which each node $v\in H$ has $deg_{in}^s (v, H)\geq k$ and $deg_{out}^t (v, H) \geq l$. Note that $s,t\in \{+, -\}$. As we have discussed before, it does not consider the internal negative edges; thus, the resultant ($l^t, k^s$)-core may contain many negative internal edges which leads to meaningless results.

Li et al.~\cite{li2018efficient} study the ($\alpha,k$)-clique problem. Specifically, given $\alpha$, $k$, and $r$, it aims to enumerate all maximal ($\alpha, k$)-cliques and find the top $r$ maximal cliques where ($\alpha,k$)-clique is a clique that satisfies negative and positive constraints.

\section{Conclusion}\label{sec:conclusion}

In this paper, we formulate and study a new core computation problem in signed networks by considering positive and negative edges simultaneously. We show that the solution of $(p,n)$-core  is not unique and NP-hard to find an exact solution; thus, we propose three algorithms to solve the problem. Finally, by using real-world and synthetic networks, we demonstrate the superiority of our algorithms. As future work, we can consider the dynamic signed network to find $(p,n)$-core by updating the existing solutions.  

\clearpage

\bibliographystyle{ACM-Reference-Format}
\bibliography{sample-base}

\end{document}